\documentclass[11pt, twoside, letter]{article}
\usepackage[margin = 1in]{geometry}
\usepackage[utf8]{inputenc}
\usepackage{amsmath, amssymb, amsfonts, cite}
\usepackage{ifthen}
\usepackage[boxed, linesnumbered]{algorithm2e}
\usepackage{comment}

\usepackage{pgf}
\usepackage{tikz}
\usetikzlibrary{arrows,automata}

\usepackage{microtype}

\usepackage{amsthm}
\usepackage{hyperref}
\usepackage{subfigure}

\newtheorem{theorem}{Theorem}[section]
\newtheorem{corollary}[theorem]{Corollary}
\newtheorem{lemma}[theorem]{Lemma}

\newtheorem{claim}[theorem]{Claim}

\theoremstyle{definition}
\newtheorem{definition}[theorem]{Definition}
\newtheorem{remark}[theorem]{Remark}

\newenvironment{fminipage}%
  {\begin{Sbox}\begin{minipage}}%
  {\end{minipage}\end{Sbox}\fbox{\TheSbox}}

\def\defeq{\stackrel{\mathrm{def}}{=}}

\def\abs#1{\left|#1  \right|}

\def\norm#1{\left\| #1 \right\|}

\newcommand\pp{\boldsymbol{\mathit{p}}}

\newcommand\rr{\boldsymbol{\mathit{r}}}

\newcommand\yy{\boldsymbol{\mathit{y}}}
\newcommand\zz{\boldsymbol{\mathit{z}}}
\newcommand\xx{\boldsymbol{\mathit{x}}}

\newcommand{\xxt}{\xx^{(t)}}
\newcommand{\yyt}{\yy^{(t)}}
\newcommand{\ppt}{\pp^{(t)}}

\newcommand\ZZ{\boldsymbol{\mathit{Z}}}

\newcommand{\zero}{\mathbf{0}}

\DeclareMathOperator*{\argmax}{arg\,max}

\begin{document}

\title{Opinion Dynamics in Networks: Convergence, Stability and Lack of Explosion\footnote{Tung Mai and Vijay V. Vazirani would like to acknowledge NSF Grant CCF-1216019. Ioannis Panageas would like to acknowledge a MIT-SUTD postdoctoral fellowship.}}

\author{
Tung Mai\\Georgia Institute of Technology\\ maitung89@gatech.edu
\and Ioannis Panageas\\MIT-SUTD\\panageasj@gmail.com
\and Vijay V. Vazirani\\Georgia Institute of Technology\\vazirani@cc.gatech.edu}

\date{}
\maketitle
\begin{abstract}
Inspired by the work of Kempe et al. \cite{kempe13}, we introduce and analyze a model on opinion formation; the update rule of our
dynamics is a simplified version of
that of \cite{kempe13}. We assume that the population is partitioned into types whose interaction pattern is specified by a graph.
Interaction leads to population mass moving from types of smaller mass to those of bigger. We show that
starting uniformly at random over all population vectors on the simplex, our dynamics converges point-wise with probability one to an independent set.
This settles an open problem of \cite{kempe13}, as applicable to our dynamics. We believe that our techniques can be used to settle the open problem for the Kempe et al. dynamics as well.

Next, we extend the model of Kempe et al. by introducing the notion of birth and death of types, with the interaction graph evolving appropriately.
Birth of types is determined by a Bernoulli process and types die when their population mass is less than a parameter $\epsilon$.
We show that if the births are infrequent, then there are long periods of ``stability" in which there is no population mass that moves.
Finally we show that even if births are frequent and ``stability" is not attained, the total number of types does not explode: it remains
logarithmic in $1/\epsilon$.

\end{abstract}
\newpage
\section{Introduction}

The birth, growth and death of political parties, organizations, social communities and product adoption groups (e.g., whether to use Windows, Mac OS or Linux) often follows common patterns, leading to the belief that the dynamics
underlying these processes has much in common.
Understanding this commonality is important for the purposes of predictability and hence has been the subject of study in mathematical social science for many years
\cite{four,eleven,twelve,seventeen,twentythree}.
In recent years, the growth of social communities on the Internet, and their increasing economic and social value,
has provided fresh impetus to this study \cite{two,three,five,nineteen}.

In this paper, we continue along these lines by building on a natural model proposed by Kempe et al. \cite{kempe13}.
Their model consists of an influence graph $G$ on $n$ vertices (types, parties) into which the entire population mass is partitioned.
Their main tenet is that individuals in smaller parties tend to get influenced by those in bigger parties \footnote{Changes in the sizes of political parties and other organizations can occur for
a multitude of possible reasons, such as changes in economic conditions,
immigration flows, wars and terrorism, and drastic changes in technology (such as the introduction
of the Internet, smart phones and social media). Studying changes due to these multitude of reasons
in a systematic quantitative manner is unrealistic. For this reason, many authors in computer science and
the social sciences have limited their work to studying the effects of relative sizes of the groups, in itself 
a key factor, e.g. see \cite{kempe13}. Following these works, our paper also takes a similar approach.}.
Individuals in the two vertices connected by an edge can interact with each other. These interactions result in individuals moving from smaller to bigger in population vertices. Kempe et al. characterize stable equilibria of this dynamics via the notion of Lyapunov stability, and they show
that under any stable equilibrium, the entire mass lies in an independent set, i.e., the population breaks into non-interacting islands. The message of this result is
clear: a population is (Lyapunov) stable, in the sense that the system does not change by much under small perturbations, only if people of different opinions do not interact. They also showed convergence to a
fixed point, not necessarily an independent set, starting from any initial population vector and influence graph. One of their main open problems was to determine whether
starting uniformly at random over all population vectors on the unit simplex, their dynamics converge with probability one to an independent set.

We first settle this open problem in the affirmative for a modification of the dynamics, which however is similar as that of Kempe et al.
in spirit in that it moves mass from smaller to bigger parties (the dynamics is defined below along with a justification).
We believe that the ideas behind our analysis can be used to settle the open problem for the dynamics of Kempe et al. as well, via
a more complicated spectral analysis of the Jacobian of the update rule of the dynamics (see Section \ref{sec:Jacobian}).

Whereas the model of Kempe et al. captures and studies the effects of migration of individuals across types in a very satisfactory manner, it is quite limited in that it does not include the birth and death of types. In this paper, we model birth and death of types. In order to arrive at realistic definitions of these
notions, we first conducted case studies of political parties in several countries.
We present below a case study on Greek politics, but similar phenomena arise in India, Spain, Italy and Holland (see Wikipedia pages).

The Siriza party in Greece provides an excellent example of birth of a party (this information is readily available in Wikipedia pages).
This party was essentially in a dormant state until the first 2012 elections in which it got 16.8\% of the vote, mostly taken away
from the Pasok party, which dropped from 43.9\% to 13.2\% in the process  (Wikipedia).
In the second election in 2012, Siriza increased its vote to 26.9\% and Pasok dropped to 12.3\%.
Finally, in 2015, Siriza increased to 36.3\% and Pasok dropped further to 4.7\%.
Another party, Potami, was formed in 2015 and got 6.1\% of the vote, again mainly from Pasok. However, in a major 2016 poll, it seems to have collapsed and is likely be absorbed by other parties.
In contrast, the KKE party in Greece, which had almost no interactions with the rest of the parties (and was like a
disconnected component), has remained between 4.5-8.5\% of the vote over the last 26 years.

Motivated by these examples, we have modeled birth and death of types in the following manner.
We model population as a continuum, as is standard in population dynamics, and time is discrete. This is the same as arXiv Version 1 of \cite{kempe13}, which is  what
we will refer to throughout this paper; the later versions study the continuous time analog.
The birth of a new type in our model is determined by a Bernoulli process, with parameter $p$. The newly born type absorbs mass from all other types via a randomized
process given by an arbitrary distribution with finite support (see Section \ref{sec:prelims.birthdeath}). After birth, the new type is connected to an arbitrary, though non-empty, set of other types.
Our model has a parameter $\epsilon$, and when the size of a type drops below $\epsilon$, it
simply dies, moving its mass equally among its neighbors.

Our rule for migration of mass, which is somewhat different from that of Kempe et al. is motivated by the following considerations. For a type $u$, $\xx_u$ will denote the fraction of population that is
of type $u$. Assume that types $u$ and $v$ have an edge, i.e., their populations interact. If so, we will assume that some individuals of the smaller type
get influenced by the larger one and move to the larger one. The question is what is a reasonable assumption on the population mass that moves.

For arriving
at the rule proposed in this paper, consider three situations. If $\xx_u = .02$ and $\xx_v = .25$, i.e., the smaller type is very small, then clearly not many people will move.
If $\xx_u = .22$ and $\xx_v = .25$, i.e., the types are approximately of the same population size, then again we expect not many people to move.
Finally, if $\xx_u = .15$ and $\xx_v = .25$, i.e., both types are reasonably big and their difference is also reasonably big, then we expect several people to move from the smaller to the bigger type. From these considerations, we propose that the amount of population mass moving from $v$ to $u$, assuming $\xx_v < \xx_u$, is given by the rule
\[ f_{v \rightarrow u}^{(t)} = \xx^{(t)}_u \xxt_v \cdot F_{uv}(\xxt_u - \xxt_v ),\] where $F_{uv}(x) = F_{vu}(x)$ is a function that captures the level of influence between $u,v$. We assume that $F_{uv}: [-1,1] \to [-1,1]$ is continuously differentiable, $F_{uv}(0)=0$ (there is no population flow between two neighboring types if they have the same fraction of population), is increasing and finally it is odd, i.e., $F_{uv}(-x) = - F_{uv}(x)$ (so that $f_{v \rightarrow u}^{(t)} = - f_{u \rightarrow v}^{(t)})$.

In this simplified setting we have made the assumption that the system is closed, i.e., that it does not get influence from outside factors (e.g., economical crisis, immigrations flows, terrorism etc).
\subsection{Our results and techniques}\label{sec:ourresults}

We first study our migration dynamics without birth and death and settle the open problem of Kempe et al., as it applies to our dynamics.
We show that the dynamics converges set-wise to a fixed point, i.e., there is a set $S$ containing only fixed points such that
the distance between the trajectory of the dynamics and $S$ goes to zero for all starting population vectors. To show this convergence result, we use a simple potential function of the population mass namely, the $\ell_2^2$ norm of the population vector, and we show that this potential is strictly increasing at each time
step (unless the dynamics is at a fixed point). Moreover, the potential is bounded, hence the result follows.
Next, we strengthen this result by showing point-wise convergence as well. The latter result is technically deeper and more difficult, since it means that every trajectory converges to a specific fixed point $\pp$. We show point-wise convergence by constructing a \textit{local} potential function that is decreasing in a small neighborhood of the limit point $\pp$. The potential function is always non-zero in that small neighborhood and is zero only at $\pp$.

Using the latter result and one of the most important theorems in dynamical systems, the Center Stable Manifold Theorem,
we prove that with probability one, under an initial population vector picked uniformly at random from the unit simplex,
our dynamics converges point-wise to a fixed point $\pp$, where the \textit{active} types $w$ in $\pp$, i.e., $w \in V(G)$ so that $\pp_w>0$, form an independent set of $G$.
This involves characterization of the linearly stable (see Section \ref{sec:definitions} for definition) fixed points and proving that the update rule of the dynamics is a local diffeomorphism\footnote{Continuously differentiable, the inverse exists and is also continuously differentiable.}. This settles the open problem of Kempe et al., mentioned in the Introduction, for our dynamics. This result is important because it allows us to perform a long-term average case analysis of the behavior of our dynamics
and make predictions.

Next, we introduce birth and death in our model. Clearly there will be no convergence in this case since new parties are created all the time.
Instead we define and study a notion of ``stability" which is different from the classical notions that appear in dynamical systems (see Section \ref{sec:definitions} for the definition of the classical notion and Definition \ref{def:stablenew} for our notion). A dynamics is $(T,d)$-stable if and only if $\forall t: \ T \leq t \leq T + d$, no population mass moves at step $t$.
We show that despite birth and death, there are arbitrarily long periods of ``stability" with high probability, for a sufficiently small $p$. Finally, we show that in the long run, with high probability, for a sufficiently large $p$, the number of types in the population will be $O(\log (1/\epsilon))$. This may seem counter-intuitive, since with a
large $p$ new types will be created often; however, since new types absorb mass from old types, the old types die frequently. In contrast,
in the short term (from the definition of $\epsilon$) we can have up to $\Theta (1/\epsilon)$ types.

Let us give an interpretation of the results of the previous paragraph in terms of political parties of certain countries (information obtained from Wikipedia).
Countries do have periods of political stability, e.g., during 1981-85, 2004-07, no new major (with more than 1\% of the vote) parties were formed in Greece,
moreover there was no substantial change in the percentage of votes won by parties in successive elections.
The parameter $\epsilon$ can be interpreted as the fraction of people that can form a party that participates in elections.
The minimum size of a party arises for organizational and legal reasons, and is $\Theta(1/Q)$, where $Q$ is the population of the country and therefore
$\epsilon$ is inversely related to population.
The message of the latter theorem is that the number of political parties grows at most as the logarithm of the population of the country, i.e. $O(\log Q)$.
The following data supports this fact. The population of Greece, Spain and India in 2015 was $1.1 e7, 4.6e7$ and $1.2e9$, respectively,
and the number of parties that participated in the general elections was $20$, $32$ and $50$, respectively.

\subsection{Related work}\label{sec:related}

As stated above, we build on the work of Kempe et al. \cite{kempe13}. They model their dynamics in a similar way, i.e., there is a flow of population for every interacting pair of types $u,v$. The flow goes from smaller to bigger types; in our case the mass is just the population of a type. One very interesting common trait between the two dynamics is that the fixed points have similar description: all types with positive mass belonging to the same connected component $C$ have the same mass. Stable fixed points also have the
same properties in both dynamics, namely they are independent sets. The update rules of the two dynamics are somewhat different; our simpler dynamics helps us in proving stronger results.

One of the most studied models is the following: there is a graph $G$ in which each vertex denotes an individual having two possible opinions. At each time step, an individual is chosen at random who next
chooses his opinion according to the majority (best response) opinion among his neighbors. This has been introduced by Galam\cite{galam} and appeared in \cite{mossel, immorlica}, where they address the question:
in which classes of graphs do individuals reach consensus. The same dynamics, but with each agent choosing his opinion according to noisy best response (the dynamics is a Markov chain) has been studied in \cite{montanari, young} and many other papers referenced therein. They give bounds for the hitting time and expected time of the consensus state (risk dominant) respectively.

Another well-known model for the dynamics of opinion formation in multiagent systems is Hegselmann-Krause \cite{krause}. Individuals are discrete entities and are modeled as points in some opinion space (e.g., real line). At every time step, each individual moves to the mass center of all the individuals within unit distance. Typical questions are related to the rate of convergence (see \cite{arnab} and references therein).
Finally, another classic model is the voter model, where there is a fixed graph $G$ among the individuals, and at every time step, a random individual selects a random neighbor and adopts his opinion \cite{holley}.
For more information on opinion formation dynamics of an individual using information learned from his neighbors, see \cite{mojackson} for a survey.

Other works, including dynamical systems that show convergence to fixed points, are \cite{PP16,ITCS15MPP,DBLP:journals/corr/LeeSJR16,PP162,mehta15,sinclairds}.
\cite{sinclairds} focuses on quadratic dynamics and they show convergence in the limit. On the other hand \cite{arora94} shows that sampling from the distribution this dynamics induces
at a given time step is PSPACE-complete.
In \cite{PP16,ITCS15MPP}, it is shown that replicator dynamics in linear congestion and 2-player coordination games converges to pure Nash equilibria, and in \cite{DBLP:journals/corr/LeeSJR16,PP162} it is shown that gradient descent converges to local minima, avoiding saddle points even in the case where the fixed points are uncountably many.

\medskip
\noindent
\textbf{Organization:} In Section \ref{sec:prelims} we describe our dynamics formally and give the necessary definitions about dynamical systems. In Section \ref{sec:convergence} we show that our dynamics without births/deaths converges with probability one to fixed points $\pp$ so that the set of types with positive population, i.e., active types, form an independent set of $G$. Finally, in Section \ref{sec:death-birth} we first show that there is no explosion in the number of types (i.e., the order never becomes $\Theta (1 / \epsilon)$) and also we perform stability analysis using our notion.

\section{Preliminaries}\label{sec:prelims}
\noindent
\textbf{Notation:} We denote the probability simplex on a set of size $n$ as $\Delta_n$. Vectors in $\mathbb{R}^n$ are denoted in boldface and $\xx_j$ denotes the $j$th coordinate of a given vector $\xx$. Time indices are denoted by superscripts. Thus, a time indexed vector $\xx$ at time $t$ is denoted as $\xxt$. We use the letters $J,\mathbb{J}$ to denote the Jacobian of a function and finally we use $f^t$ to denote the composition of $f$ by itself $t$ times.

\subsection{Migration dynamics}
\label{sec:migration}
Let $G = (V,E)$ be an undirected graph on $n$ vertices (which we also call types), and let $N_v$ denote the set of neighbors of $v$ in $G$. During the whole dynamical process, each
vertex $v$ has a non-negative population mass representing the fraction of the population of type $v$. We consider a discrete-time process and let
$\xx^{(t)}_v$ denote the mass of $v$ time step $t$. It follows that the condition
\[ \sum_{v \in V(G)} \xx^{(t)}_v = 1, \]
must be maintained for all $t$, i.e., $\xxt \in \Delta_n$\footnote{Recall that $\Delta_n$ denotes the simplex of size $|V(G)|=n$.} for all $t \in \mathbb{N}$.

Additionally, we consider a dynamical migration rule where the population can move along edges of $G$ at each step. The movement at step $t$
is determined by $\xx^{(t)}$. Specifically, for $uv \in E(G)$, the amount of mass moving from $v$
to $u$ at step $t$ is given by
\[ f^{(t)}_{v \rightarrow u} =  \xx^{(t)}_u \xxt_v F_{uv}(\xxt_u - \xxt_v ).\]
For all $uv \in E(G)$ we assume that $F_{uv}: [-1,1] \to [-1,1]$ is a continuously differentiable function such that:
\begin{enumerate}
	\item  $F_{uv}(0)=0$ (there is no population flow between two neighboring types if they have the same fraction of population),
	\item  $F_{uv}$ is increasing (the larger $\xx_u -  \xx_v$, the more population moving from $v$ to $u$ ),
	\item  $F_{uv}$ is odd i.e., $F_{uv}(-x) = - F_{uv}(x)$ (so that $f_{v \rightarrow u}^{(t)} = - f_{u \rightarrow v}^{(t)})$.
\end{enumerate}
It can be easily derived from the assumptions that $F_{uv}(x) \geq F_{uv}(0)=0 $ for $x \geq 0$ and $F'_{uv}(-x) = F'_{uv}(x)$ for $x \in [-1,1]$, where $F'_{uv}$ denotes the derivative of $F_{uv}.$
Note that $f^{(t)}_{v \rightarrow u} > 0$ implies that population is moving from $v$ to $u$, and $f^{(t)}_{v \rightarrow u} < 0$ implies that population
is moving in the other direction.
The update rule for the population of type $u$ can be written as
\begin{align}\label{eq:dynamics}
\xx^{(t+1)}_u 	&=  \xx^{(t)}_u + \sum_{v \in N_u} f^{(t)}_{v \rightarrow u}  \\
				&= \xx^{(t)}_u + \sum_{v \in N_u}  \xxt_u \xxt_v F_{uv}(\xxt_u - \xxt_v ).
\end{align}
We denote the update rule of the dynamics as $g : \Delta_n \to \Delta_n$, i.e., we have that
\begin{equation*}
\xx^{(t+1)} = g(\xx^{(t)}).
\end{equation*}
Therefore it holds that $\xxt = g^t (\xx^{(0)})$, where $g^t$ denotes the composition of $g$ by itself $t$ times. It is easy to see $g$ is well-defined for $\sup_{x\in [-1,1]} |F_{uv}(x)| \leq 1$ for all $uv \in E(G)$, in the sense that if $\xx^{(t)} \in \Delta_n$ then $\xx^{(t+1)} \in \Delta_n$. This is true
since for all $u$ we get (using induction, i.e., $\xxt \in \Delta_n$)
\begin{align*}
\xx^{(t+1)}_u &= \xxt_u + \sum_{v \in N_u}  \xxt_u \xxt_v F_{uv}(\xxt_u - \xxt_v ) \\&\geq \xxt_u - \sum_{v \in N_u} \xxt_u \xxt_v
\\&\geq \xxt_u-\xxt_u(1-\xxt_u) \geq 0,
\end{align*}
moreover it holds
\begin{align*}
\xx^{(t+1)}_u &= \xxt_u + \sum_{v \in N_u} \xxt_u \xxt_v F_{uv}(\xxt_u - \xxt_v ) \\&\leq \xxt_u + \sum_{v \in N_u} \xxt_u \xxt_v
\\&\leq \xxt_u+\xxt_u(1-\xxt_u) \\&\leq \xxt_u+1-\xxt_u =1,
\end{align*}
and also $\sum_u \xx^{(t+1)}_u = \sum_u \xxt _u = 1$ (the other terms cancel out).

\subsection{Birth and death of types}\label{sec:prelims.birthdeath}

Political parties or social communities don't tend to survive once their size becomes ``small'' and hence there is a
need to incorporate death of parties in our model. We will define a global parameter $\epsilon$ in our model.
When the population mass of a type $v$ becomes smaller than some fixed value $\epsilon$, we consider it to be dead and move its mass arbitrarily to existing types. Formally, if $\xxt_v \leq \epsilon$ then $\xxt_v \leftarrow 0$ and $\xxt_u  \leftarrow \xxt_u + \xxt_v/ \abs{N_v}$ for all $u \in N_v$.
Also, vertex $v$ is removed and edges are added arbitrarily on its neighbors to ensure connectivity of the resulting graph.

\begin{remark} It is not hard to see that the maximum number of types is $1/\epsilon$ (by definition). We say that we have explosion in the number of types if they are of $\Theta(1/\epsilon)$. In Theorem \ref{thm:bound} we show that in the long run, the number of types is much smaller, it is $O(\log (1/\epsilon))$) with high probability.
\end{remark}


Every so often, new political opinions emerge and like-minded people move from the existing parties to create a new
party, which then follows the normal dynamics to either survive or die out. To model birth of new types, at each time step, with probability $p$, we create a new type $v$ such that $v$ takes a portion of mass from each existing type independently. The amount of mass going to $v$ from each $u$ follows an arbitrary distribution in the range $[\beta_{\min},\beta_{\max}]$ . Specifically, let $\ZZ_u \sim \mathcal{D}$ where $\mathcal{D}$ is a distribution with support $[\beta_{\min},\beta_{\max}]$, the amount of mass going from $u$ to $v$ is
\[  \ZZ_u \xx_u.\]
We connect $v$ to the existing graph arbitrarily such that it remains connected.

Additionally, we make a small change to the migration dynamics defined in Section~\ref{sec:migration} to make it more realistic. Our tenet is that population mass
migrates from smaller to bigger types because of influence. However, if the two types are of approximately the same size, the difference is size is not discernible
and hence migration should not happen. To incorporate this, we introduce a new parameter $\delta > 0$ and if $\abs{ \xx_u - \xx_v} \leq \delta$, we
assume that no population moves from $u$ to $v$.

Finally, each step of the dynamics consists of there phases in the following order: 
\begin{enumerate}
\item Migration: the dynamics follows the update rule from Section~\ref{sec:migration}.
\item Birth: with probability $p$, a new type $v$ is created and takes mass from the existing types.
\item Death: a type with mass smaller than $\epsilon$ dies out and move its mass to the existing types.
\end{enumerate}

\begin{remark} For any different order of phases, all proofs in the paper still go through with minimal changes.
\end{remark}

\subsection{Definitions and basics}\label{sec:definitions}
A recurrence relation of the form $\xx^{(t+1)} = f(\xx^{(t)})$ is a discrete time dynamical system, with update rule $f:\mathcal{S} \to \mathcal{S}$ (for our purposes, the set $\mathcal{S}$ is $\Delta_n$).
The point $\zz$ is called a \textit{fixed point} or \textit{equilibrium} of $f$ if $f(\zz) = \zz$.
A fixed point $\zz$ is called \textit{Lyapunov stable} (or just stable) if for every $\varepsilon > 0$, there exists a $\zeta = \zeta(\varepsilon) > 0$ such
that for all $\xx$ with $\norm{\xx-\zz} < \zeta$ we have that $\norm{f^k(\xx)-\zz} < \varepsilon$ for every $k \geq 0$.
We call a fixed point $\zz$ \textit{linearly stable} if, for the Jacobian $J(\zz)$ of $f$, it holds that its spectral radius is at most one.
It is true that if a fixed point $\zz$ is stable then it is linearly stable but the converse does not hold in general \cite{perko}.
A sequence $(f^t(\xx^{(0)}))_{t \in \mathbb{N}}$ is called a \textit{trajectory} of the dynamics with $\xx^{(0)}$ as starting point.
A common technique to show that a dynamical system converges to a fixed point is to construct a function $P : \Delta_m \to \mathbb{R}$ such that $P(f(\xx)) > P(\xx)$ unless $\xx$ is a fixed point.
We call $P$ a \textit{potential} or \textit{Lyapunov} function.

\section{Convergence to independent sets almost surely}\label{sec:convergence}
In this section we prove that the deterministic dynamics (assuming no death/birth of types, namely the graph $G$ remains fixed) converges point-wise to %
fixed points $\pp$ where $\{v: \pp_v>0\}$ (set of active types) is an independent set of the graph $G$, with probability one assuming that the
starting point $\xx^{(0)}$ follows an atomless distribution with support in $\Delta_n$. To do that, we show that for all starting points $\xx^{(0)}$, the
dynamics converges point-wise to fixed points. Moreover we prove that the
update rule of the dynamics is a diffeomorphism and that the linearly stable fixed points $\pp$ of the dynamics satisfy the fact that the set of active types in %
$\pp$ is an independent set of $G$. Finally, our main claim of the section follows by using a well-known theorem in dynamical systems, called
Center-Stable Manifold theorem.

\medskip
\noindent
\textbf{Structure of fixed points.} The fixed points of the dynamics (\ref{eq:dynamics}) are vectors $\pp$ such that for each $uv \in E(G)$, at least one of the following conditions must hold:
\begin{equation*}
1. \;\pp_v = \pp_u,\; 2.\; \pp_v = 0,\; 3.\; \pp_u = 0.
\end{equation*}
Therefore, for each fixed point $\pp$, the set of active types (types with non-zero population mass) with respect to $\pp$ must form a set of
connected components such that all types in each component have the same population mass.
We first prove that the dynamics converges point-wise to fixed points.

\subsection{Point-wise convergence}
First we consider the following function
\[ \Phi(\xx) = \sum_v \xx_v^2 \]
and state the following lemma on $\Phi$.

\begin{lemma}[\textbf{Lyapunov (potential) function}]\label{lem:increasing}
Let $\xx$ be a vector with $\xx_u > \xx_v$. Let $\yy$ be another vector such that $\yy_v = \xx_v - d$, $\yy_u = \xx_u + d$ for some $0 < d \leq \xx_v$ and  $\yy_z = \xx_z$ for all $z \neq u,v$. Then
\[ \Phi(\xx) < \Phi(\yy). \]	
\end{lemma}

\begin{proof}
By the definition of $\Phi$,
\begin{align*}
	\Phi(\yy) 	&= \yy_v^2 + \yy_u^2 + \sum_{z \neq u,v} \yy_z^2  \\
				&= \left(\xx_v - d \right)^2 + \left(\xx_u + d \right)^2 + \sum_{z \neq u,v} \yy_z^2 \\
				&= \xx_v^2 - 2 d \xx_v + d^2 + \xx_u^2 - 2 d \xx_u + d^2 + \sum_{z \neq u,v} \xx_z^2 \\
				&= \Phi(\xx) + 2d(\xx_u - \xx_v) + 2d^2 \\
				&> \Phi(\xx).
\end{align*}
The inequality follows because $d > 0$ and $\xx_u > \xx_v$.
\end{proof}

If we think of $\xx$ as a population vector, Lemma~\ref{lem:increasing} implies that $\Phi(\xx)$ increases if population is moving from a smaller type to a bigger type.

\begin{theorem}[\textbf{Set-wise convergence}] \label{thm:convergence}
 $\Phi(\xxt)$ is strictly increasing along every nontrivial trajectory, i.e., $\Phi(\xx^{t+1}) = \Phi(g(\xxt)) \geq \Phi(\xxt)$ with equality only when $\xxt$ is a fixed point. As a corollary, the dynamics converges to fixed points (set-wise convergence).  
\end{theorem}

\begin{proof}
First we prove that the dynamical process converges to a set of fixed points by showing that $\Phi(\xxt)$ is strictly increasing as $t$ grows. The idea is %
breaking a migration step from $\xxt$ to $\xx^{(t+1)}$ into multiple steps such that each small step only involves migration between two types.
Moreover, in each small step, population is moving from a smaller type to a bigger type. Lemma~\ref{lem:increasing} guarantees that $\Phi$ is strictly
increasing in every small step, and thus strictly increasing in the combined step from $\xxt$ to $\xx^{(t+1)}$.

Let $D$ be the directed graph representing the migration movement from $\xxt$ to $\xx^{(t+1)}$. Formally, for each edge $uv \in E(G)$ we direct $uv$ in both %
directions and let $f_{v \rightarrow u} = \max ( 0, f^{(t)}_{v \rightarrow u} )$. Define the following process on $D$:
\begin{enumerate}
	\item  If there exists a directed path $v \rightarrow u \rightarrow z$ of length 3 in $D$ such that $f_{v \rightarrow u}$ and $f_{u \rightarrow z}$ are both positive, we make the following modification to the flow in $D$.
	
	Let $\Delta = \min (f_{v \rightarrow u} , f_{u \rightarrow z} )$, and
	\begin{align*}
		f_{v \rightarrow u} &\leftarrow f_{v \rightarrow u} - \Delta \\
		f_{u \rightarrow z} &\leftarrow f_{u \rightarrow z} - \Delta \\
		f_{v \rightarrow z} &\leftarrow f_{v \rightarrow z} + \Delta.
	\end{align*}
	\item  Keep repeating the previous step until there is no path of length 3 carrying positive flow.
\end{enumerate}
The above process must terminate since the function
\[ \sum_{v \rightarrow u} f_{v \rightarrow u} (\xxt_u - \xxt_v)^2 \]
strictly increases in each modification and is bounded above.
At the end of the process, there is no path of length 3 carrying positive flow. In other words, each type can not
have both flows coming into it and flows coming out of it.
Furthermore, it is easy to see that the net flow $\sum_{v \in V(G)} f_{v \rightarrow u}$ at each type $u$ is preserved and $f_{v \rightarrow u} > 0$ only if $\xxt_u > \xxt_v$.
We can break a migration step from $\xxt$ to $\xx^{(t+1)}$ into multiple migrations such that each small migration corresponds to a flow $f_{v \rightarrow u} > 0$ at the end of the process.
It follows that in each small migration, population is moving from one smaller type to one bigger type.

\medskip
To finish the proof we proceed in a standard manner as follows: Let $\Omega \subset \Delta_n$ be the set of limit points of a trajectory (orbit) $\xxt$. $\Phi(\xxt)$ is increasing with respect to time $t$ by above and so, because $\Phi$ is bounded on $\Delta_n$, $\Phi(\xxt)$ converges as $t\to \infty$  to $\Phi^{*}= \sup_t\{\Phi(\xxt)\}$. By continuity of $\Phi$ we get that $\Phi(\pp)= \lim_{t\to\infty} \Phi(\xxt) = \Phi^*$ for all $\pp \in \Omega$, namely $\Phi$ is constant on $\Omega$. Also $\ppt = \lim_{k \to \infty} \xx^{(t_k + t)}$ as $k \to \infty $ for some sequence of times $\{t_i\}$, with $\pp^{(0)} = \lim_{k \to \infty} \xx^{(t_k)}$ and $\pp^{(0)} \in \Omega$. Therefore $\ppt$ lies in $\Omega$ for all $t \in \mathbb{N}$, i.e., $\Omega$ is invariant. Thus, since $\pp^{(0)} \in \Omega$ and the orbit $\ppt$ lies in $\Omega$, we get that $\Phi(\ppt) = \Phi^*$ on the orbit. But $\Phi$ is strictly increasing except on fixed points and so $\Omega$ consists entirely of fixed points.
\end{proof}

\begin{theorem}[\textbf{Point-wise convergence}]\label{thm:point-wise}
The dynamics converges point-wise to fixed points.
\end{theorem}
\begin{proof}
We first construct a local potential function $\Psi$ such that $\Psi(\xx,\pp)$ is strictly decreasing in some small neighborhood of a limit point (fixed point) $\pp$. Formally we initially prove the following:
\begin{claim}[\textbf{Local Lyapunov function}]\label{cl:local} Let $\pp$ be a limit point (which will be a fixed point of the dynamics by Theorem \ref{thm:convergence}) of trajectory $\yyt$, and $\Psi(\xx,\pp)$ be the following function
\[ \Psi(\xx,\pp) = \sum_{v:\pp_v > 0} ( \pp_v - \xx_v  ). \] There exists a small $\varepsilon>0$ such that if $\norm{\yyt-\pp}_1\leq \varepsilon$ then $\Psi(\yy^{(t+1)},\pp) \leq \Psi(\yyt,\pp).$
\end{claim}
\begin{proof}[Proof of Claim \ref{cl:local}]
We know that the set of active types in $\pp$ must form a set of connected components such that all types in each component have the same population mass with respect to $\pp$. Let $C$ be one such component and let $\delta(C) = \{v: v \not \in C, u \in C, uv \in E(G) \}$. We choose $\varepsilon$ to be so small so that $\yyt_u > \yyt_v$ for all $u \in C$ and $v \in \delta(C)$, because $\yyt_v<\varepsilon$ since $\pp_v=0$ (thus is arbitrarily close to zero) for $v \in \delta(C)$. Therefore, the net flow into $C$ must be non-negative, thus $\Psi(g(\yyt),\pp)\leq \Psi(\yyt,\pp)$.

Additionally, we have that $\yyt_u \leq \pp_u$ for all $u \in C$. Suppose otherwise, then there exists a type $w \defeq w(t) \in C$ with $\yyt_w > \pp_w$, which is a contradiction. To see why, consider $w(t')$ to be $\argmax_{z \in C} \yy^{(t')}_z$, then it should hold that $\yy^{(t')}_{w(t')} \geq \pp_{w(t')} + s$ for some constant $s>0$ independent of $t'$ and $t' \geq t$, because $\yy^{(t')}_{w(t')}$ is increasing and therefore $\pp$ cannot be a limit point of the trajectory $\yyt$. Hence, $\Psi(\yyt,\pp)$ must be non-negative and only equal to zero when $\yyt = \pp$ (i).
\end{proof}

To finish the proof of the theorem, if $\pp$ is a limit point of $\yyt$, there exists an increasing sequence of times $\{t_{i}\}$, with $t_{n} \to \infty$ and $\yy^ {(t_{n})} \to \pp$. We consider $\varepsilon '$ such that the set $S = \{\xx : \xx \leq \pp \textrm{ and }\Psi(\xx,\pp)\leq \varepsilon '\}$ is inside $B = \norm{\xx-\pp}_1<\varepsilon$ where $\varepsilon$ is from Claim \ref{cl:local} about the local potential. Since $\yy^{( t_{n})} \to \pp$, consider a time $t_{N}$ where $\yy^{ (t_{N})}$ is inside $S$. From Claim \ref{cl:local} we get that if $\yyt \in B$ then $\Psi(\yy^{(t+1)},\pp)\leq \Psi(\yyt,\pp)$ (and also $\yyt \in S$), thus $\Psi(\yyt) \leq \Psi(\yy^{ (t_{N})},\pp) \leq \varepsilon '$ and $\yyt\leq \pp$ for all $t \geq t_{N}$ (namely the orbit remains in $S$; we use Claim \ref{cl:local} inductively). Therefore $\Psi(\xxt,\pp)$ is decreasing in $S$ and since $\Psi(\yy ^{(t_{n})},\pp) \to \Psi(\pp,\pp) =0$, it follows that $\Psi(\yyt,\pp) \to 0$ as $t \to \infty$. Hence $\yyt \to \pp$ as $t \to \infty$ using (i).
\end{proof}

\subsection{Diffeomorphism and stability analysis via Jacobian}
\label{sec:Jacobian}
In this section we compute the Jacobian $J$ of $g$ and then perform spectral analysis on $J$. The Jacobian of $g$ is the following:
\begin{align*}
\frac{\partial g_u}{ \partial \xx_u} = J_{u,u} &= 1 +  \sum_{v \in N_u}\xx_v \left[F_{uv}(\xx_u - \xx_v)+\xx_u F'_{uv}(\xx_u - \xx_v)\right],\\
\frac{\partial g_u}{\partial \xx_v} = J_{u,v} &=  \xx_u \left[F_{uv}(\xx_u-\xx_v)-\xx_vF'_{uv}(\xx_u-\xx_v)\right] \textrm{ if }uv \in E(G) \textrm{ else }0.
\end{align*}
\begin{lemma}[\textbf{Local Diffeomorphism}]\label{lem:invertible} The Jacobian is invertible on the subspace $\sum_v \xx_v=1$, for $\sup_{x\in [-1,1]}|F_{uv}(x)|<\frac{1}{2}$ for each $uv \in E(G)$. Moreover, $g$ is a local diffeomorphism in a neighborhood of $\Delta_n$.
\end{lemma}
\begin{proof} First we have that $\sum_{v \in N_u}\xx_v F_{uv}(\xx_u - \xx_v)+\sum_{v \in N_u}\xx_u\xx_v F'_{uv}(\xx_u - \xx_v)\geq -\sum_{v \in N_u}  \xx_v \geq -1$ and hence $J_{u,u}>0$ for all $u$ and $\xx \in \Delta_n$. The first inequality comes from the fact that $F'_{uv} (x)\geq 0$ ($F_{uv}$ is increasing) and $F_{uv}(x) > -1/2 > -1 $. Additionally, we get that
\begin{align*}
|J_{u,u}| - \sum_{v \neq u}|J_{v,u}| &= 1 + \sum_{v \in N_u}\xx_v \left[F_{uv}(\xx_u - \xx_v)+\xx_u F'_{uv}(\xx_u - \xx_v)\right] - \sum_{J_{v,u}>0}J_{vu} + \sum_{J_{v,u}<0}J_{vu}
\\& = 1 + 2 \sum_{v \in N_u, J_{v,u}>0}\xx_v \left[F_{uv}(\xx_u - \xx_v)+\xx_u F'_{uv}(\xx_u - \xx_v)\right]
\\&\geq 1 - 2 \sum_{v \in N_u, J_{v,u}>0}\xx_v F_{uv}(\xx_u - \xx_v)
\\& >  1 - \sum_{v \in N_u} \xx_v \geq 0,
\end{align*}
where we used the fact that $\frac{1}{2}>F_{uv}(x)>-\frac{1}{2}, F_{uv}(-x) = -F_{uv}(x)$ and that $\sum_{v \in N_u} \xx_v \leq 1$. Therefore we conclude that $J^{\top}$ is diagonally dominant.

Finally, assume that $J$ is not invertible, then there exists a nonzero vector $\yy$ so that $J^{\top}\yy=\zero$ (ii). We consider the index type with the maximum absolute value in $\yy$, say $w$. Hence we have $|\yy_w| \geq |\yy_{v}|$ for all $v \in V(G)$. Finally,
using (ii) we have that $J_{w,w}\yy_w = - \sum_{v \neq w} J_{v,w}\yy_v$, thus $J_{w,w} \leq \sum_{v \neq w}|J_{w,v}| \frac{|\yy_v|}{|\yy_w|}\leq \sum_{v \neq w}|J_{w,v}|$ (first inequality is triangle inequality and second comes from assumption on $w$). We reach a contradiction because we showed before that $J^{\top}$ is diagonally dominant.

\medskip
Therefore, $J(\xx)$ is invertible for all $\xx \in \Delta_n$. Moreover, from inverse function theorem we get the claim that the update rule of the dynamics is a local diffeomorphism in a neighborhood of $\Delta_n$.
\end{proof}

\begin{lemma}[\textbf{Linearly stable fixed point $\Rightarrow$ independent set}]\label{lem:indsetstable} Let $\pp$ be a fixed point so that there exists a connected component $C$, $|C|>1$ with $\pp_v> 0$ same population mass for all $v\in C$. Then the Jacobian at $\pp$ has an eigenvalue with absolute value greater than one.
\end{lemma}
\begin{proof} The Jacobian at $\pp$ has equations:
\begin{enumerate}
\item Assume $\pp_u=0$ then \begin{align*}
J_{u,u} &= 1 +  \sum_{v \in N_u} \pp_v F_{uv}(-\pp_v)= 1 -  \sum_{v \in N_u} \pp_v F_{uv}(\pp_v),\\
J_{u,v} &=  0 \;\textrm{ for all }v\neq u.
\end{align*}
\item Assume $\pp_u>0$ then
\begin{equation*}
J_{u,u} = 1 + \sum_{v \in N_u, \pp_v >0}F'_{uv}(0)\pp_v^2.
\end{equation*}
\end{enumerate}
For all $v \in V(G)$ so that $\pp_v=0$, it follows that $J_{v,v'}$ is nonzero only when $v'=v$ (diagonal entry) and thus the corresponding eigenvalue of $J$ is  $1 -  \sum_{v' \in N_v} \pp_{v'}F_{vv'}(\pp_{v'}) \leq 1$ with left eigenvector $(0,...,0,\underbrace{1}_{v\textrm{-th}},0,...,0)$ (it is clear that $\pp_{v'}F_{vv'}(\pp_{v'})\geq 0$ since $F_{vv'}(x) \geq F_{vv'}(0)=0$ for $x\geq 0$). Hence the characteristic polynomial of $J$ at $\pp$ is equal to $$\prod_{v : \pp_{v}=0}\left(\lambda
-\left[1 -  \sum_{v' \in N_v} \pp_{v'}F_{vv'} (\pp_{v'})\right] \right) \times \textrm{det}(\lambda I -\mathbb{J}),$$ where $\mathbb{J}$ corresponds to $J$ at $\pp$ by deleting rows corresponding to types $v$ such that $\pp_v=0$ and $I$ the identity matrix. So it suffices to prove that $\mathbb{J}$ has an eigenvalue with absolute value greater than one.

Assume $\mathbb{J}$ has size $l \times l$, in other words the number of active types is $l$ in $\pp$. Every type $v$ that has no active neighbors satisfies $\mathbb{J}_{v,v}=1$ and every type $v$ that has at least one active neighbor satisfies $\mathbb{J}_{v,v}=1+ \sum_{v'\in N_v, \pp_{v'}>0}F'_{vv'}(0)\pp_{v'}^2>1$. Therefore $\textrm{trace}(\mathbb{J})>l$ by assumption on $\pp$. Hence the sum of the eigenvalues of $\mathbb{J}$ is greater than $l$, thus there exists an eigenvalue with absolute size greater than 1.
\end{proof}

\subsection{Center-stable manifold and average case analysis}
In this section we prove our first main result, Corollary \ref{thm:probabilityone}, which is a consequence of the following theorem:
\begin{theorem}\label{thm:zero} Assume that $\max_{x \in [-1,1]}|F_{uv}(x)| < 1/2$ for all $uv \in E(G)$. 
The set of points $\xx \in \Delta_n$ such that dynamics \ref{eq:dynamics} starting at $\xx$ converges to a fixed point $\pp$ whose active types do not form an independent set of $G$ has measure zero.
\end{theorem}

To prove Theorem \ref{thm:zero}, we are going to use arguably one of the most important theorems in dynamical systems, called {\em Center Stable Manifold Theorem}:
\begin{theorem}[\textbf{Center-stable Manifold Theorem} \cite{shub}]\label{thm:manifold}
Let $\pp$ be a fixed point for the $C^r$ local diffeomorphism $f: U \to \mathbb{R}^m$ where $U \subset \mathbb{R}^m$ is an open neighborhood of $\pp$ in $\mathbb{R}^m$ and $r \geq 1$. Let $E^s \oplus E^c \oplus E^u$ be the invariant splitting
of $\mathbb{R}^m$ into generalized eigenspaces of the Jacobian of g, $J(\pp)$ corresponding to
eigenvalues of absolute value less than one, equal to one, and greater than one. To the $J(\pp)$ invariant subspace $E^s\oplus
E^c$ there is an associated local $f$ invariant $C^r$ embedded disc $W^{sc}_{loc}$ tangent to the linear subspace at $\pp$ and a ball $B$ around $\pp$ such that:
\begin{equation} f(W^{sc}_{loc}) \cap B \subset W^{sc}_{loc}.\textrm{  If } f^m(\xx) \in B \textrm{ for all }m \geq 0,
\textrm{ then }\xx \in W^{sc}_{loc}.
\end{equation}
\end{theorem}

Since an $n$-dimensional simplex $\Delta_n$ in $\mathbb{R}^n$ has dimension $n-1$, we need to take a projection of the domain space ($\sum_v \xx_v=1$) and accordingly redefine the map $g$. Let $\xx$ be a point mass in $\Delta_n$. Let $u$ be a fixed type and define $h:\mathbb{R}^n \to \mathbb{R}^{n-1}$ so that we exclude the variable $\xx_u$ from $\xx$, i.e., $h(\xx) = \xx_{-u}$. We substitute the variable
$\xx_{u}$ with $1-\sum_{v\neq u} \xx_v$ and let $g'$ be the resulting update rule of the dynamics $g'(\xx_{-u}) = g(\xx)$. The following lemma gives a relation between the eigenvalues of the Jacobians of functions $g$ and $g'$.
\begin{lemma}\label{lem:sameeigenvalues} Let $J,J'$ be the Jacobian of $g,g'$ respectively. Let $\lambda$ be an eigenvalue of $J$ so that $\lambda$ does not correspond to left eigenvector $(1,\ldots,1)$ (with eigenvalue 1). Then $J'$ has also $\lambda$ as an eigenvalue.
\end{lemma}
\begin{proof}
By chain rule, the equations of $J'$ are as follows:
\begin{equation*}
\frac{\partial g'_v}{\partial \xx_w} = J'_{v,w}(\xx_{-u}) = J_{v,w}(\xx_{-u},1-\sum_{v\neq u}\xx_v)- J_{v,u}(\xx_{-u},1-\sum_{v\neq u}\xx_v).
\end{equation*}
Assume $\lambda$ is associated with left eigenvector $\rr \defeq (\rr_1,...,\rr_{n-1},\rr_n)$ (we label the types with numbers $1,\ldots,n$ with $u$ taking index $n$). We claim that $\lambda$ is an eigenvalue of $J$ with right eigenvector $\rr' \defeq (\rr_1-\rr_n,...,\rr_{n-1}-\rr_n)$. First it is easy to see that
\begin{align*}
\sum_{j=1}^{n-1} J'_{ji}(\rr_j-\rr_n) &=
\sum_{j=1}^{n-1} (J_{ji} - J_{jn})(\rr_j-\rr_n) \\&= \sum_{j=1}^{n-1}J_{ji}\rr_{j}- \sum_{j=1}^{n-1}J_{jn}\rr_{j}- \sum_{j=1}^{n-1}J_{ji}\rr_n+\sum_{j=1}^{n-1}J_{jn}\rr_n.
\end{align*}
Since $\rr$ is a left eigenvector, we get that $\sum_{j=1}^n J_{ji}\rr_j = \lambda \rr_i$ for all $i \in [n]$ and also it holds that $\sum_{j=1}^n J_{ji}=1$ for all $i \in [n]$. Therefore
\begin{align*}
\sum_{j=1}^{n-1} J'_{ji}(\rr_j-\rr_n) &= (\lambda \rr_i - J_{ni}\rr_n) - (\lambda \rr_n - J_{nn}\rr_n) - (1-J_{ni})\rr_n + (1-J_{nn})\rr_n
\\& = \lambda (\rr_i-\rr_n),
\end{align*}
namely $\rr' J' = \lambda \rr'$ and the lemma follows.
\end{proof}
Before we proceed with the proof of Theorem \ref{thm:zero}, we state the following which is a corollary of Lemmas \ref{lem:invertible}, \ref{lem:indsetstable} and \ref{lem:sameeigenvalues} and also uses classic properties for determinants of matrices.
\begin{corollary}\label{cor:projected} Let $\pp$ be a fixed so that the active types are not an independent set in $G$, then $J'$ at $h(\pp)$ has an eigenvalue with absolute value greater than one. Additionally, The Jacobian $J'$ of $g'$ is invertible in $h(\Delta_n)$ and as a result $g'$ is a local diffeomorphism in a neighborhood of $h(\Delta_n)$.
\end{corollary}
\begin{proof} If $\pp$ is a fixed point where the active types are not an independent set in $G$, then by Lemma $\ref{lem:indsetstable}$ we get that $J$ at $\pp$ has an eigenvalue with absolute value greater than one, hence using Lemma \ref{lem:sameeigenvalues} it follows that $J'$ at $h(\pp)$ has an eigenvalue with absolute value greater than one.

Let $B$ be the resulting matrix if we add all the first $n-1$ rows to the $n$-th row and then subtract the $n$-th column from all other columns in matrix $J$. It is clear that $\textrm{det}(B) = \textrm{det}(J) \neq 0$ (determinant not zero since $J$ is invertible from Lemma \ref{lem:invertible}). Additionally, the last row of $B$ is all 0's and $B_{nn}=1$, so $\textrm{det}(B) = \textrm{det}(B')$ where $B'$ is the resulting matrix if we delete from $B$ last row,column. But $B' = J'$, hence $0 \neq \textrm{det}(J) = \textrm{det}(B) = \textrm{det}(J')$ and thus $J'$ is invertible, therefore $g'$ is a local diffeomorphism in a neighborhood of $h(\Delta_n)$ (by Inverse function theorem).
\end{proof}

\begin{proof}[Proof of Theorem \ref{thm:zero}] Let $\pp$ be a fixed point of function $g(\xx)$ so that the set of active types is not an independent set. We consider the projected fixed point $\pp' \defeq h(\pp)$ of function $g'$. Then $\pp'$ is a linearly unstable fixed point.
Let $B_{\pp'}$ be the (open) ball (in the set $\mathbb{R}^{n-1}$) that is derived from center-stable manifold theorem. We consider the union of these balls $$A = \cup _{\pp'}B_{\pp'}.$$

Due to Lindel\H{o}f's lemma \ref{thm:lindelof} stated in the appendix , we can find a countable subcover for $A$, i.e., there exist fixed points $\pp'_1,\pp'_2,\dots$ such that $A = \cup _{m=1}^{\infty}B_{\pp'_{m}}$.
Starting from a point $\xx' \in \Delta' \defeq h(\Delta_n) $, there must exist a $t_{0}$ and $m$ so that $g'\; ^{t}(\xx) \in B_{\pp'_{m}}$ for all $t \geq t_{0}$ because of Theorem \ref{thm:point-wise}, i.e., because the dynamics (\ref{eq:dynamics}) converges point-wise. From center-stable manifold theorem we get that $g'\; ^t(\pp') \in W_{loc}^{sc}(\pp'_m) \cap \Delta'$ where we used the fact that $g'(\Delta') \subseteq \Delta'$ (the population vector is always in simplex, see Section \ref{sec:migration}), namely the trajectory remains in $\Delta'$ for all times \footnote{$W_{loc}^{sc}(\pp'_m)$ denotes the center stable manifold of fixed point $\pp'_m$}.

By setting $D_{1}(\pp'_m) = g'\; ^{-1}(W_{loc}^{sc}(\pp'_m) \cap \Delta')$ and $D_{i+1}(\pp'_m) = g'\; ^{-1}(D_{i}(\pp'_m) \cap \Delta')$ we get that $\xx' \in D_{t}(\pp'_m)$ for all $t \geq t_0$.
Hence the set of initial points in $\Delta'$ so that dynamics converges to a fixed point $\pp'$ so that the set of active types is not an independent set of $G$ is a subset of
\begin{equation}
P =  \cup_{m=1}^{\infty} \cup_{t=0}^{\infty} D_{t}(\pp'_m)).
\end{equation}

Since $\pp'_m$ is linearly unstable fixed point (for all $m$), the Jacobian $J'$ has an eigenvalue greater than 1, and therefore the dimension of $W_{loc}^{sc}(\pp'_m)$ is at most $n-2$. Thus, the set $W_{loc}^{sc}(\pp'_m) \cap \Delta'$ has Lebesgue measure zero in $\mathbb{R}^{n-1}$.
Finally since $g'$ is a local diffeomorphism, $g'\; ^{-1}$ is locally Lipschitz (see \cite{perko} p.71). $g'\; ^{-1}$ preserves the null-sets (by Lemma \ref{lem:lips} that appears in the appendix) and hence (by induction) $D_{i}(\pp'_m)$ has measure zero for all $i$. Thereby we get that $P$ is a countable union of measure zero sets, i.e., is measure zero as well.
\end{proof}

\begin{corollary}[\textbf{Convergence to independent sets}]\label{thm:probabilityone} Suppose that $\max_{x \in [-1,1]}|F_{uv}(x)| < 1/2$ for all $uv \in E(G)$. If the initial mass vector $\xx^{(0)} \in \Delta_n$ is chosen from an atomless distribution, then the dynamics converges point-wise with probability $1$ to a point $\pp$ so that the active types form an independent set in $G$.
\end{corollary}
\begin{proof}
The proof comes from Theorem \ref{thm:point-wise} and Theorem \ref{thm:zero}. From Theorem \ref{thm:point-wise} we have that $\lim_{t \to \infty} \xxt$ exists for all $\xx^{(0)} \in \Delta_n$ and from Theorem \ref{thm:zero} we get the probability that the dynamics converges to fixed points where the active types are not an independent set is zero. Hence the dynamics converges to fixed points where the active types is an independent set, with probability one.
\end{proof}

Corollary \ref{thm:probabilityone} is illustrated in Figure \ref{fig:example} for the case of a 3-path and a triangle. As shown in the figure, if the initial condition is chosen uniformly at random from a point in the simplex, the dynamics converges to an independent set with probability one.

\section{Stability and bound on the number of types}
\label{sec:death-birth}

In this section we consider dynamical systems with migration, death and birth and prove two probabilistic statements on stability and number of types. 
Recall that in this settings, the dynamics at each step contains three phases in order: migration $\rightarrow$ birth $\rightarrow$ death.
The following direct application of Chernoff's bound is used intensively to attain probabilistic guarantees.

\begin{lemma}
\label{lem:chernoff}
Consider a period of $t$ steps.
\begin{enumerate}
	\item There are at least $tp/2$ births with probability at least $1- e^{-tp/8}$.
	\item There are at most $3tp/2$ births with probability at least $1- e^{-tp/6}$.
\end{enumerate}
\end{lemma}
\begin{proof} Let $X_i = 1$ if there is a birth at step $i$, and $X_i = 0$ otherwise. The number of births in $k$ step is
\[ X = \sum_{i=1}^k X_i\]
It follows that $\text{E}[X] = tp$. From Chernoff's bound,
\[ \Pr \left( X \leq tp/2 \right) \leq e^{-tp/8}.  \]
In other words,
\[ \Pr \left( X \geq tp/2 \right) \geq 1- e^{-tp/8}.  \]
Applying Chernoff's bound again, we have
\[ \Pr \left( X \geq 3tp/2 \right) \leq e^{-tp/6}, \]
and
\[ \Pr \left( X \leq 3tp/2 \right) \geq 1- e^{-tp/6}. \]
\end{proof}

\subsection{Stability}

We define the notion of stability and give a stability result for a dynamical system involving migration, death and birth. For the rest of the paper, we denote by $\alpha_{\min} = \min_{uv \in E(G), x \in [-1,1]} F'_{uv}(x)$ and $\alpha_{\max} = \max_{uv \in E(G), x \in [-1,1]} F'_{uv}(x)$. It can be seen easily that for each $uv \in E(G)$,
\[\alpha_{\min} x \leq  F_{uv}(x) - F_{uv}(0) \leq \alpha_{\max} x. \]

\begin{definition}[\textbf{$(T,d)$-Stable dynamics}]\label{def:stablenew} A dynamics is $(T,d)$-stable if and only if $\forall T \leq t \leq T + d$, no population mass moves in the migration phase at step $t$.
\end{definition}


We state the following two lemmas whose proofs come from the definition of $\Phi$.

\begin{lemma} \label{lem:increasingBound}
If the dynamics is not $(t,0)$-stable, the migration phase at time $t$ increases $\Phi$ by at least $2\alpha_{\min} \epsilon \delta^3$.
\end{lemma}
\begin{proof} By Theorem~\ref{thm:convergence}, we know that  $\Phi$ is strictly increasing in each migration when the dynamics is not at a fixed point. Moreover, it is easy to see that the increase is minimized when there is only one edge $uv$ carrying population flow. From the proof of Lemma~\ref{lem:increasing}, the additive increase in $\Phi$ is at least
\[ 2 d \abs{\xx_u - \xx_v}, \]
where $d$ is the amount of mass moving along $uv$. Without loss of generality, we may assume that $\xx_u > \xx_v$. Since $\xx_v \geq \epsilon$ and $\xx_u - \xx_v \geq \delta$,
\[ d = \xx_u \xx_v F_{vu}\left( \xx_u - \xx_v\right) \geq \alpha_{\min} \xx_u \xx_v \left( \xx_u - \xx_v\right) \geq \alpha_{\min} \epsilon \delta^2.\]
It follows that
\[ 2 d \abs{\xx_u - \xx_v} \geq 2 \alpha_{\min} \epsilon \delta^3. \]
\end{proof}

\begin{lemma}
\label{lem:decreasingBound}	
Each birth can decrease $\Phi$ by at most 2$\beta_{\max}$.
\end{lemma}

\begin{proof} Recall that
\[ \Phi(\xx) = \sum_{v} \xx_v^2. \]
The potential after a new type is created is
\[ \sum_{v} \xx_v^2(1 - \ZZ_v)^2 +  \sum_v \ZZ_v^2 \xx_v^2. \]
Therefore, the net decrease in potential is
\begin{align*}
\Delta \Phi 	&= \sum_{v} \xx_v^2 - \sum_{v} \xx_v^2(1 - \ZZ_v)^2 -  \sum_v \ZZ_v^2 \xx_v^2\\
		&= \sum_{v} \xx_v^2 - \sum_{v} \xx_v^2(1 -2\ZZ_v + \ZZ_v^2) -  \sum_v \ZZ_v^2 \xx_v^2 \\
		&=  2 \sum_{v} \ZZ_v \xx_v^2 - \sum_{v} \ZZ_v^2 \xx_v^2  -  \sum_v \ZZ_v^2 \xx_v^2 \\
		&\leq 2 \sum_{v} \ZZ_v \xx_v^2 \leq 2 \beta_{\max} \sum_{v} \xx_v^2 \\
		&\leq 2\beta_{\max} \left( \sum_{v} \xx_v \right)^2 = 2\beta_{\max}.
\end{align*}
\end{proof}
 With the two above lemmas, we can give a theorem on the stability of the dynamics. At a high level, it says that if the probability of a new type emerging is small enough, then as time goes on, there must be a long period period such that there is no migration. 

\begin{theorem}[\textbf{``Stable" for long enough}]\label{thm:stability}
Let $p <  \min \left( \frac{\epsilon \delta^3 \alpha_{\min}}{3 \beta_{\max} } , \frac{2}{3}\right)$ and  $t > \frac{1}{\epsilon \delta^3 \alpha_{\min} - 3p\beta_{\max}}$. With probability at least $1- e^{-tp/6}$, the dynamics is $ \left(T, \frac{1}{3p} \right)$-stable for some $T \leq t$.
\end{theorem}

\begin{proof} Consider a period of $t$ steps. By Lemma~\ref{lem:chernoff}, there are at most $3tp/2$ births in the period with probability at least $1- e^{-tp/6}$. Note that $p <2/3$ guarantees that $3tp/2 < t$. In the migration phases of the period, $\Phi$ can either increase if there is a migration or remain unchanged otherwise.

Assume that $\Phi$ increases in more than $t/2$ migration phases. By Lemma~\ref{lem:increasingBound}, the amount of increase in potential due to migrations is at least
\[ \frac{t}{2}2 \alpha_{\min} \epsilon \delta^3 =  t \alpha_{\min} \epsilon \delta^3. \]
Since there are at most $3tp/2$ births, Lemma~\ref{lem:decreasingBound} guarantees that the amount of decrease in potential due to births is at most
\[ \frac{3tp}{2} 2 \beta_{\max} = 3tp \beta_{\max}. \]
Therefore, the net increase of $\Phi$ is least
\[t \alpha_{\min} \epsilon \delta^3 - 3tp \beta_{\max} = t(\alpha_{\min} \epsilon \delta^3   - 3p \beta_{\max} ).\]
Since $t > \frac{1}{\epsilon \delta^3 \alpha_{\min} - 3p \beta_{\max} }$, the net increase in $\Phi$ is greater than 1, which is a contradiction.

It follows that $\Phi$ cannot increase in more than $t/2$ migration phases, and must remain unchanged in at least $t/2$ migration phases. Note that if the dynamics is $(t',0)$-stable for some $t'$, it will be $(t',d)$-stable until the next birth at $t'' = t'+d+1$. Since there are at most $3tp/2$ births, there must be no migration in a period of
\[ \frac{ t/2}{3tp/2} = \frac{1}{3p}\]
consecutive steps.
\end{proof}

\subsection{Bound on the number of types}

In this section we investigate a behavior of the dynamics following a long period of time. Specifically, we show that after a large number of steps, the number of types can not be too high. 
Our goal is to prove the following theorem:

\begin{theorem} [\textbf{Lack of explosion}]\label{thm:bound}
Let $\alpha_{\max} \leq p/512$ and $t \geq (16/p) \log ^2(1/\epsilon) $. The dynamics at step $t$ has at most $ 72 \log (1/\epsilon)$ types with probability at least $1 - 3\epsilon$. 
\end{theorem}

First we give the following lemma, which says that if the number of types is large enough, then after a fixed period of time, it will decrease by a factor of roughly 2.
\begin{lemma} \label{lem:numDecreasing}
Let $\alpha_{\max} \leq p/512$ and $k$ be the number of types at step $t_0$. If $k \geq 48 \log (1/\epsilon)$, with probability at least $1 - 2 \epsilon^2 $, 
the number of types at step $t_0 +  (16/p) \log (1/\epsilon) $ is at most $k/2+ 24 \log (1/\epsilon)$.
\end{lemma}

\begin{proof}
Without loss of generality, we may assume $t_0 = 0$.
Our goal is to show that with high probability, at least half of the original types (types at step $0$) will die from step $0$ to step $(16/p) \log (1/\epsilon) $.
From Lemma~\ref{lem:chernoff}, with probability at least $1 - 2\epsilon^2$, there are at least $8 \log(1 / \epsilon)$ and at most  $24 \log(1 / \epsilon)$ births in that period.
Therefore, we may assume that the number of births in the period is between $8 \log(1 / \epsilon)$ and $24 \log(1 / \epsilon)$.
 Since the number of new types created is at most $24 \log (1/\epsilon)$, if at least half of the original types die out from step $0$ to step $(16/p) \log (1/\epsilon)$, the lemma will immediately follow. Assume that the number of original types dying in the period is less than $k/2$ for the sake of contradiction. It follows that the number of remaining original types at the end of the period is at least $k/2$.

Since $24 \log (1/ \epsilon) \leq k/2$, the total number of types dying out in the period is at most $k/2 + 24 \log (1/ \epsilon) \leq k$. Therefore, the total mass that can be added to the remaining types from the dead types is at most $k \epsilon$. It follows that on average, a remaining original type receive at most $2 \epsilon$ from the dead types. Markov's inequality says that at most half of the remaining original types can receive more than $4 \epsilon$ from the dead types. Therefore, at most $k/4$ original types receive more than $4\epsilon$ from the dead types. Since the average mass of an original type is $1/k$, Markov's inequality again guarantees that the number of original types having mass greater than $4/k$ is at most $k/4$. Combining the above two reasons, we can conclude that at least $k/2$ original types have initial mass less than $4/k$ and receive less than $4\epsilon$ from the dead types in the whole period.

We will prove that those types can not remain at the end of the period. The idea is to bound the total increase in the masses of those types, and argue that after at least $8 \log (1/ \epsilon)$ births, their masses will all be less than $\epsilon$.

Consider an original type $v$ such that $\xx^{(0)}_v \leq 4/k$ and $v$ receives less than $4\epsilon$ from the dead types from step 0 to step $(16/p) \log (1/\epsilon)$. Recall that the change in mass of $v$ at step $t$ is
\begin{align*}
\Delta \xx^{(t)}_v 	&=  \sum_{u \in N_v} \xxt_v \xxt_u F_{vu}\left( \xxt_v - \xxt_u \right) \\
			&\leq \sum_{u \in N_v} \alpha_{\max} \xxt_v \xxt_u \left( \xxt_v - \xxt_u \right) \\
			&\leq  \sum_{u \in N_v} \alpha_{\max} \left( \xxt_v \right)^3 \\
			&\leq  \left( k + 24 \log (1/ \epsilon)  \right) \alpha_{\max} \left( \xxt_v \right)^3 \\
			&<  2k  \alpha_{\max} \left( \xxt_v \right)^3.
\end{align*}
It follows that
\[ \xx^{(t+1)}_v \leq \xxt_v + 2k\alpha_{\max} \left( \xxt_v \right)^3. \]

Ignoring the effect of births from step 0 to step $t < (16/p) \log (1/ \epsilon)$, we claim that
\[ \xx^{(t)}_v \leq \frac{8}{k} + 16 \alpha_{\max} kt \left( \frac{8}{k} \right)^3. \]

We will prove the above claim by induction. We may assume that $v$ receives all the mass from dead types at the beginning, since this assumption can only increase $\xx^{(t)}_v $. We have
\[\xx^{(0)}_v \leq \frac{4}{k} + 4 \epsilon \leq \frac{8}{k},\]
where the last inequality follows since $k \leq 1/\epsilon$. Hence, the base case is satisfied. Suppose that the claim is true for $t$, then we have
\begin{align*}
\xx^{(t+1)}_v 	&\leq \xxt_v + 2k\alpha_{\max} \left( \xxt_v \right)^3 \\
			&\leq  \frac{8}{k} + 16 \alpha_{\max} kt \left( \frac{8}{k} \right)^3 + 2k\alpha_{\max} \left( \frac{8}{k} + 16 \alpha_{\max} kt \left( \frac{8}{k} \right)^3 \right)^3 \\
			&\leq  \frac{8}{k} + 16 \alpha_{\max} kt \left( \frac{8}{k} \right)^3 +  2k\alpha_{\max} \left( 2 \frac{8}{k} \right)^3 \\
			&= \frac{8}{k} + 16 \alpha_{\max} k(t+1) \left( \frac{8}{k} \right)^3.
\end{align*}
The last inequality follows because
\[ k \geq  48 \log (1/\epsilon) \geq 48 \log (1/\epsilon) (512 \alpha_{\max} / p) > 1024 \alpha_{\max} t \]
and thus,
\[ \frac{8}{k} =  k \frac{8}{k^2} > 1024 \alpha_{\max} t \frac{8}{k^2} = 16\alpha_{\max} t k \left( \frac{8}{k} \right)^3.  \]

Now the mass of $v$ decreases by at least a multiplicative factor of $1- \beta_{\min}$ at each birth. We may assume that the decreases on $\xx_v$ happen after the increases since this assumption can only increase the bound on $\xx_v$. We have
\[ \xx^{(t)}_v \leq \left( \frac{8}{k} + 16 \alpha_{\max} kt \left( \frac{8}{k} \right)^3 \right) (1- \beta_{\min})^B, \]
where $B$ is the number of births in the period of $t$ steps. Setting $t = (16/p) \log ( 1 / \epsilon)$ and $B = 8 \log ( 1 / \epsilon) $ gives
\begin{align*}
\xx^{(t)}_v &\leq \left( \frac{8}{k} + 16 \alpha_{\max} k (16/p) \log( 1 / \epsilon) \left( \frac{8}{k} \right)^3 \right) (1- \beta_{\min})^{8 \log ( 1 / \epsilon)} \\
&\leq \left( \frac{8}{k} + 256   \log (1 / \epsilon) \frac{1}{k^2} \right) \epsilon \\
&= \frac{16 \epsilon}{k}.
\end{align*}
Therefore, $\xx^{(t)}_v < \epsilon$ as desired.
\end{proof}

Now we can prove Theorem~\ref{thm:bound}.

\begin{proof}[Proof of Theorem~\ref{thm:bound}]
We only consider the last $(16/p) \log ^2(1/\epsilon)$ steps and assume that $t = (16/p) \log ^2(1/\epsilon)$.

We call a period of $(16/p) \log(1/\epsilon)$ steps a \emph{decreasing period} if it satisfies the condition in Lemma~\ref{lem:numDecreasing}, i.e., if the number of types $k$ at the beginning of the period is at least $48 \log (1/\epsilon)$, and the number of types at the end of the period is at most $k/2+ 24 \log (1/\epsilon)$.
Construct a set $P$ of periods of length $(16/p) \log(1/\epsilon)$ as follows:
\begin{enumerate}
	\item Start with $t'=0$.
	\item Repeat the following step until $t' = t$:
	\begin{enumerate}
		\item If $t' + (16/p) \log(1/\epsilon) \leq t$ and the number of types at $t'$ is at least $48 \log (1/\epsilon)$, let $i$ be the period from $t'$ to $t' + (16/p) \log(1/\epsilon)$, and add $i$ to $P$. Update $t' \leftarrow t' + (16/p) \log(1/\epsilon)$.
		\item Else update $t' \leftarrow t' + 1$.
	\end{enumerate}
\end{enumerate}

Assume that all periods in $P$ are decreasing periods. By Lemma~\ref{lem:numDecreasing}, the probability of such an outcome occurring is at least
\[ ( 1 - 2\epsilon^2 )^{\log(1/\epsilon)} \geq  1 - 2\epsilon^2\log(1/\epsilon) \geq 1 - 2 \epsilon. \]

First, we prove that the number of types must become smaller than $48 \log (1/\epsilon)$ at some step of the dynamics. Assume that the number of types is at least $48 \log (1/\epsilon)$ throughout the dynamics. Since the dynamics has $(16/p) \log ^2(1/\epsilon)$ steps and each period has $(16/p) \log(1/\epsilon)$ steps, there are $\log (1/\epsilon)$ periods in $P$. Let $T(n)$ be the number of types after the $n$-th period is added to $P$. Since all periods in $P$ are decreasing periods, we have the recurrence
\[T(n) \leq  {T(n-1)}/{2} + 24 \log (1/\epsilon).\]
Note that $T(0) \leq 1/\epsilon$ in the base case. The above recurrence gives $T \left(\log(1/\epsilon) \right) < 48 \log(1/\epsilon)$. In other words, the number of types becomes smaller than $48 \log(1/\epsilon)$ at the end of the dynamics, which is a contradiction.

Now, we know that the number of types becomes smaller than $48 \log (1/\epsilon)$ at some step. If it remains smaller than $48 \log (1/\epsilon)$, the theorem holds trivially. Therefore, we may assume that it reaches $48 \log (1/\epsilon)$ at some later step. We consider two cases: 
\begin{enumerate}
\item There are at least $(16/p) \log(1/\epsilon)$ subsequent steps. According to our assumption, that period of $(16/p) \log(1/\epsilon)$ steps must be a decreasing period. Therefore, the number of types is at most $48 \log (1/\epsilon)$ after the period.
\item There are less than $(16/p) \log(1/\epsilon)$ subsequent steps after the number of types reaches $48 \log (1/\epsilon)$. By Lemma~\ref{lem:chernoff}, the probability that in those remaining steps, there are at most $24 \log (1 / \epsilon)$ births is at least
\[ 1 - e^{-\frac{16\log (1 / \epsilon)}{6} } \geq 1 - \epsilon.\]
By union bound, the probability that there are at most $72 \log (1 / \epsilon)$ at the end of the dynamics is at least
\[1 - (2\epsilon + \epsilon) = 1- 3 \epsilon. \]
\end{enumerate}


\end{proof}

Theorems \ref{thm:stability} and \ref{thm:bound} are summarized in Figure \ref{fig:open}.

\section{Conclusion}
In this paper, we introduce and analyze a model on opinion formation. In the first part, the dynamics is deterministic and we don't have either deaths or births of types. We show that the dynamics in this case converges point-wise to fixed points $\pp$, where the set of active types in $\pp$ forms an independent set of $G$. After introducing births and deaths of types, we show that with high probability in the long run we reach a state in which there is no movement of population mass for a long period of time (aka ``stable"). We also show that the number of types is logarithmic in $1/\epsilon$, where $\epsilon$ is the size of a type at which it dies.

A host of novel questions arise from this model and there is much scope for future work:
\begin{itemize}
	\item \textbf{Rate of convergence (without births and deaths):} How fast does our migration dynamics converge point-wise to fixed points $\pp$ for different choices of functions $F_{uv}$?
	How does the structure of $G$ influence the time needed for convergence?
	Assuming $F_{uv}(z) = a_{uv} z$ (linear functions), do the values of $\alpha_{uv}$'s affect the speed of convergence?
	\item \textbf{Average case analysis:} Theorem \ref{thm:zero} gives qualitative information for the behavior of the dynamics assuming no births and deaths of types. However, it is not clear which independent sets are more likely to occur if we start at random in the simplex. Assuming $F_{uv}(z) = a_{uv} z$ (linear functions) or $F_{uv}(z) = a_{uv} z^3$ etc (cubic), do the values of $\alpha_{uv}$'s affect the likelihood of the linearly stable fixed points\footnote{this likelihood of a fixed point is called region of attraction.}?
\item \textbf{Understanding the behavior of the dynamics:} From Figure \ref{fig:open} we see that when $p$ lies in the regime $(\Theta(\epsilon), \Theta(1))$, we don't understand the behavior of the system, e.g., we don't know if we have explosion in the number of types (i.e., having $\Theta(1/\epsilon)$ types) in the long run. Moreover, we don't know if the system reaches ``stability" (in our notion).	
\item \textbf{Relaxing the notion of stability:} If we relax the notion of ``stability" so that $\gamma$-fraction of the population is allowed to move ($1-\gamma$ does not move for some $0<\gamma<1$), can we give better guarantees in Theorem \ref{thm:stability}?
\end{itemize}

\bibliographystyle{plain}
\bibliography{bibliography}

\newpage
\appendix
\section{Appendix}

The following theorem holds for every separable metric space, i.e., every metric space that contains a countable, dense subset. In particular, we use this theorem for $\mathbb{R}^{n-1}$ where $n$ is the number of types in the proof of Theorem \ref{thm:zero}.
\begin{theorem} [\textbf{Lindel\H{o}f's lemma} \cite{kelley}] \label{thm:lindelof} For every open cover there is a countable subcover.
\end{theorem}

The following lemma is used in Theorem \ref{thm:zero} when we argue that the set of initial population vectors so that the dynamics converges to fixed points with an unstable direction, has measure zero. It roughly states that if a function $h$ is locally Lipschitz, then it preserves the measure zero sets (measure zero sets are mapped to measure zero sets).
\begin{lemma}[\textbf{Null-set preserving}, Appendix of \cite{ITCS15MPP}]\label{lem:lips} Let $h: \mathcal{S} \to \mathbb{R}^m$ be a locally Lipschitz function with $\mathcal{S} \subseteq \mathbb{R}^m$, then $h$ is null-set preserving, i.e., for $E \subset \mathcal{S}$ if $E$ has measure zero then $h(E)$ has also measure zero.
\end{lemma}
\begin{proof} Let $B_{\gamma}$ be an open ball such that $\norm{h(\vec{y}) - h(\vec{x})} \leq K_{\gamma} \norm{\vec{y}-\vec{x}}$
for all $\vec{x},\vec{y} \in B_{\gamma}$. We consider the union $\cup_{\gamma}B_{\gamma}$ which cover $\mathbb{R}^m$ by the
assumption that $h$ is locally Lipschitz. By Lindel\H{o}f's lemma we have a countable subcover, i.e., $\cup_{i=1}^{\infty}B_{i}$. Let
$E_{i} = E \cap B_{i}$. We will prove that $h(E_{i})$ has measure zero. Fix an $\epsilon >0$. Since $E_{i} \subset E$, we have that
$E_{i}$ has measure zero, hence we can find a countable cover of open balls $C_{1},C_{2},... $ for $E_{i}$, namely $E_{i} \subset
\cup_{j=1}^{\infty}C_{j}$ so that $C_{j} \subset B_{i}$ for all $j$ and also $\sum_{j=1}^{\infty} \mu(C_{j}) <
\frac{\epsilon}{K_{i}^m}$. Since $E_{i} \subset \cup_{j=1}^{\infty}C_{j}$ we get that $h(E_{i}) \subset \cup_{j=1}^{\infty}h(C_{j})$,
namely $h(C_{1}),h(C_{2}),...$ cover $h(E_{i})$ and also $h(C_{j}) \subset h(B_{i})$ for all $j$. Assuming that ball $C_{j} \equiv
B(\vec{x},r)$ (center $\vec{x}$ and radius $r$) then it is clear that $h(C_{j}) \subset B(h(\vec{x}),K_{i} r)$ ($h$ maps the
center $\vec{x}$ to $h(\vec{x})$ and the radius $r$ to $K_{i}r$ because of Lipschitz assumption). But $\mu(B(h(\vec{x}),K_{i}
r)) = K_{i}^m \mu(B(\vec{x}, r)) = K_{i}^m \mu(C_{j})$, therefore $\mu(h(C_{j})) \leq K_{i}^m \mu(C_{j})$ and so we conclude that
$$\mu(h(E_{i})) \leq \sum_{j=1}^{\infty}\mu(h(C_{j})) \leq K_{i}^m \sum_{j=1}^{\infty}\mu(C_{j}) < \epsilon$$ Since $\epsilon$ was
arbitrary, it follows that $\mu(h(E_{i})) =0$. To finish the proof, observe that $h(E) = \cup_{i=1}^{\infty} h(E_{i})$ therefore
$\mu(h(E)) \leq \sum_{i=1}^{\infty} \mu(h(E_{i})) =0$.  \end{proof}

\newpage

\begin{figure}
\section{Figures}

\begin{minipage}{1.0\textwidth}
\centering     
\subfigure[The region with ``C" corresponds to the initial population vectors so that the dynamics converges to the fixed point where all the population is of type $C$. The region ``A+B" corresponds to the initial population masses so that the dynamics converges to a fixed point where part of the population is of type $A$ and the rest of type $B$.]{\label{fig:path}\includegraphics[width=0.45\linewidth,height=0.45\linewidth]{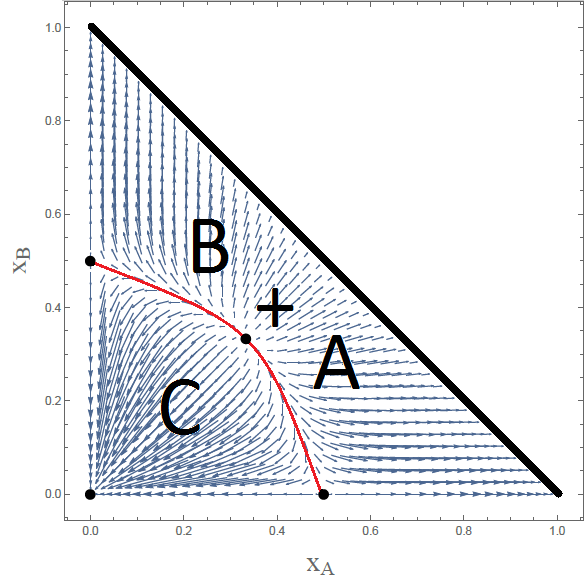}}
\;\;\;\subfigure[Each region ``A", ``B", ``C" corresponds to the initial population vectors so that the dynamics converges to all the population being of type $A,B,C$ respectively. It is easy to see that an initial vector $(\xx_A,\xx_B,\xx_C)$ converges to the fixed point where all population is of type $\argmax_{i\in\{A,B,C\}}\xx_i$. In case of ties, the limit population is split equally among the tied types (symmetry).]{\label{fig:triangle}\includegraphics[width=0.45\linewidth,height=0.45\linewidth]{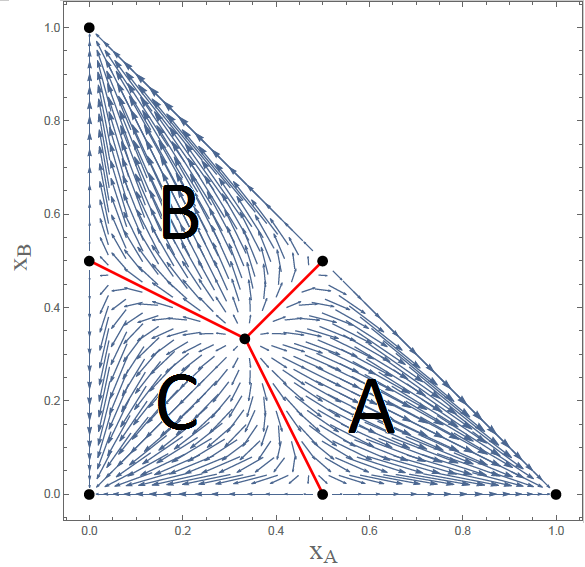}}
\caption{Migration dynamics phase portrait for path and triangle of 3 types $A,B,C$ respectively and for $F_{uv}(x) = 0.5x$ for all $uv \in E(G)$. The black points and the line correspond to the fixed points. $\xx_A,\xx_B$ correspond to the frequencies of people that are of type $A,B$. We omit $\xx_C$ since $\xx_C = 1-\xx_A-\xx_B$.}
 \label{fig:example}
\end{minipage}

\medskip
\medskip
\medskip
\medskip

\begin{minipage}{1.0\textwidth}

\includegraphics[height=0.12\linewidth]{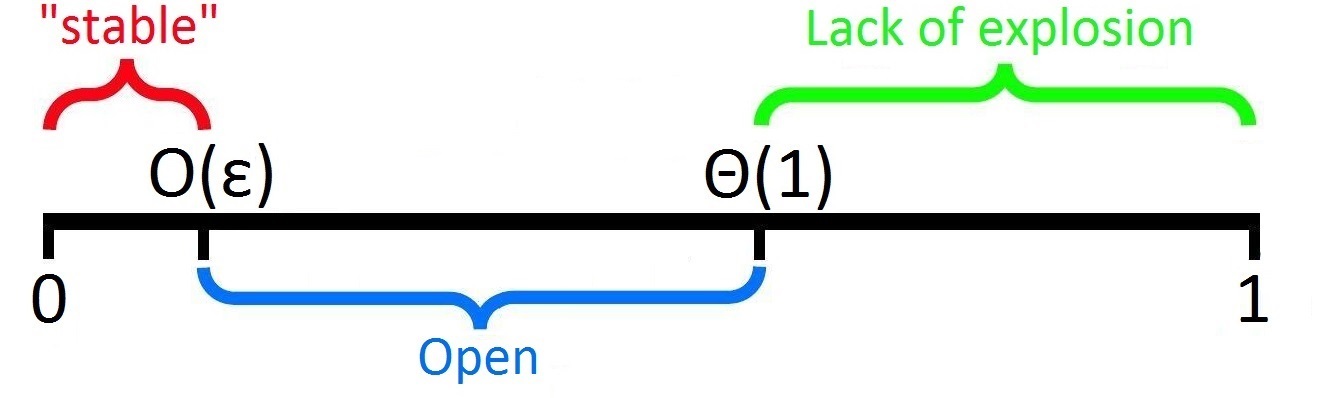}
\centering     
\caption{In the red interval where $p$ is $O(\epsilon)$ we have that in the long run the system reaches a state where there is no migration of population mass for a long period of time. In the green interval where $p$ is $\Theta(1)$ we have that there is no explosion in the number of types, namely the number is at most $O(\log (1 /\epsilon))$ in the long term. In the blue interval, we don't have a qualitative or quantitative characterization of the system.}
\label{fig:open}
\end{minipage}

\end{figure}

\end{document}